\RequirePackage{fix-cm}
\documentclass[smallextended]{svjour3}       
\smartqed  
\usepackage{graphicx}
\usepackage{amssymb}
\usepackage{mathrsfs}
\usepackage{color}
\begin{document}

\title{The 2-adic complexity of Yu-Gong sequences with interleaved structure and optimal autocorrelation magnitude \thanks{This work is financially supported by the National Natural Science Foundation of China (No. 61902429), the Fundamental Research Funds for the Central Universities (No. 20CX05012A), and the Major Scientific and Technological Projects of CNPC under Grant ZD2019-183-008.\\
$\ast$ The corresponding author}}



\author{Yuhua Sun$^{\ast}$ \and Tongjiang Yan \and Qiuyan Wang}


\institute{Yuhua Sun, Tongjiang Yan \at
               College of Sciences,
China University of Petroleum,
Qingdao 266555,
Shandong, China;
\at
Provincial Key Laboratory of Applied Mathematics, Putian University, Putian,
Fujian 351100,
China;\\
Qiuyan Wang \at
School of Computer Science and Technology, \\
Tiangong University, Tianjin 300387,China
\\
          }
\date{Received: date / Accepted: date}

\maketitle

\begin{abstract}
In 2008, a class of binary sequences of period $N=4(2^k-1)(2^k+1)$ with optimal autocorrelation magnitude has been presented by Yu and Gong based on an $m$-sequence, the perfect sequence $(0,1,1,1)$ of period $4$ and interleaving technique. In this paper, we study the 2-adic complexity of these sequences. Our result shows that it is larger than $N-2\lceil\mathrm{log}_2N\rceil+4$ (which is far larger than $N/2$) and could attain the maximum value $N$ if suitable parameters are chosen, i.e., the 2-adic complexity of this class of interleaved sequences is large enough to resist the Rational Approximation Algorithm.


\end{abstract}

\section{Introduction}
\label{section 1}

Since the interleaved structure of sequences was introduced by Gong in \cite{Gong-1}, several classes of binary sequences with this form have been constructed and were proved to have so many good pseudo-random properties, such as low autocorrelation, large linear complexity. For example, in 2010, Tang and Gong constructed three classes of sequences with optimal autocorrelation value/magnitude using Legendre sequences, twin-prime sequences and a generalized GMW sequence, respectively \cite{Tang-Gong}, which were showed by Li and Tang to have large linear complexity \cite{LiNian-1}. In quick succession, Tang and Ding presented two more general constructions which include constructions in \cite{Tang-Gong} as special cases and gave more sequences with optimal autocorrelation and large linear complexity \cite{Tang-Ding}. Later, Yan et al. also gave a generalized version for the constructions in \cite{Tang-Gong} and put forward a sufficient and necessary condition for an interleaved sequence to have optimal autocorrelation \cite{Yan-Information S}. What's more exciting is that these sequences have also been proved to have large
2-adic complexity by Xiong et al. \cite{Xiong-1,Xiong-2} and Hu \cite{Hu-1} using different methods respectively. Moreover, Su et al. constructed another class of sequences with optimal autocorrelation magnitude combining interleaved structure and Ding-Helleseth-Lam sequences \cite{SuWei-1} and these sequences have also been shown to have large linear complexity by Fan \cite{Fan-1} and large 2-adic complexity by Sun et al. \cite{SunYuhua-1} and Yang et al. \cite{Yang-1}.

Note that each of the above mentioned sequences can be described as an interleaved form $s=I(s_{1},s_2,s_{3},s_{4})$, i.e., the sequence $s$ is obtained by concatenating the successive rows of the matrix $I(s_{1},s_2,s_{3},s_{4})$, in which each column is a periodic sequence $s_i$, $1\leq i\leq4$. In fact, Yu and Gong also presented another description method of an interleaved structure by an indicator sequence (see Construction 1) \cite{Yu-Gong-1}. Using this description, Yu and Gong represented an ADS (almost difference set) sequence of period $4v$ in \cite{Arasu} as a $v\times4$ interleaved structure and a product sequence of period $4v$ in \cite{Luke} as a $4\times v$ interleaved structure respectively, which provides us a new understanding for the two sequence structures. Not only that, they also discovered another new classes of sequences with optimal autocorrelation magnitude and large linear complexity using binary $m$-sequences as the indictor sequences, which we call Yu-Gong sequences. However, the 2-adic complexity of this class of sequences has not been studied yet as far as we know.

In this paper, using the method of Hu \cite{Hu-1}, we investigate the 2-adic complexity of a Yu-Gong sequence with an $m$-sequence as its indicator sequence, which is proved to be lower bounded by $N-2\lceil\mathrm{log}_2N\rceil+4$ ($\gg \frac{N}{2}$) and could attain the maximum value $N$ if suitable parameters are chosen, where $N$ is the period of the sequence.

The rest of the paper is organized as follows. Some notations and definitions are introduced in section 2. We describe the generalized construction and the definition of a Yu-Gong sequence in Section 3. In Section 4, we point out a very interesting and useful law of the autocorrelation values of a Yu-Gong sequence. Using the method of Hu and the law of the autocorrelation distribution of a Yu-Gong sequence, we derive a lower bound on the 2-adic complexity of this sequence in Section 5.

\section{Preliminaries}
The following symbols will be used throughout the whole paper.
\begin{itemize}
       \item [(1)] $\mathbb{Z}_{N}$ is a ring of integers modulo $N$ and $\mathbb{Z}_{N}^{+}=\{t\in \mathbb{Z}_{N}|t\neq0\}$.
       \item [(2)] $\mathbb{F}_q$ is a finite field with $q$ elements.
       \item [(3)] For positive integers $n$ and $m$ satisfying $m\mid n$, the trace function $Tr_m^n(x)$ from $\mathbb{F}_{2^n}$ to $\mathbb{F}_{2^m}$ is defined by
       $$Tr_m^n(x)=x+x^{2^{m}}+\cdots+x^{2^{m(\frac{n}{m}-1)}},\ \ x\in  \mathbb{F}_{2^n}.$$
\end{itemize}

Let $\mathbf{s}=(s_{0},s_{1},\cdots,s_{N-1})$ be a binary sequence of period $N$. Then the autocorrelation of $\mathbf{s}$
is given by
\begin{equation}
AC_{\mathbf{s}}(\tau)=\sum_{t=0}^{N-1}(-1)^{s_{t}+s_{t+\tau}},\ \ \ \ 0\leq\tau\leq N-1,\nonumber
\end{equation}
where $\tau$ is called a phase shift of the sequence $\mathbf{s}$ and $t+\tau$ is computed modulo $N$. The sequence $\mathbf{s}$ is called to have optimal autocorrelation if $AC_\mathbf{s}(\tau)$ satisfies the following:
\begin{itemize}
  \item [(1)] $AC_{\mathbf{s}}(\tau)\in\{N,1,-3\}$ for $N\equiv1\ (\mathrm{mod}\ 4)$ or
  \item [(2)] $AC_{\mathbf{s}}(\tau)\in\{N,2,-2\}$ for $N\equiv2\ (\mathrm{mod}\ 4)$ or
  \item [(3)] $AC_{\mathbf{s}}(\tau)\in\{N,-1\}$ for $N\equiv3\ (\mathrm{mod}\ 4)$ or
  \item [(4)] $AC_{\mathbf{s}}(\tau)\in\{N,0,-4\}$ or $\{N,0,4\}$ for $N\equiv0\ (\mathrm{mod}\ 4)$
\end{itemize}
for all $\tau$'s. Specially, the case (3) is called to have ideal two-level autocorrelation. Additionally, for $N\equiv0\ (\mathrm{mod}\ 4)$ and all $\tau$'s, it is called to have perfect autocorrelation if $AC_{\mathbf{s}}(\tau)\in\{N,0\}$ and optimal autocorrelation magnitude if $AC_{\mathbf{s}}(\tau)\in\{N,0,4,-4\}$. So far, the sequence $(0,1,1,1)$ of period 4 is a uniquely known binary sequence with perfect autocorrelation in the sense of cyclic equivalence. Hence, it is often used to construct new sequences with good correlation and Yu-Gong sequence discussed in this paper is one of the applications.

Denote $S(x)=\sum\limits_{i=0}^{N-1}s(i)x^i\in \mathbb{Z} [x]$ and suppose
\begin{equation}
\frac{S(2)}{2^N-1}=\frac{\sum\limits_{i=0}^{N-1}s(i)2^i}{2^N-1}=\frac{e}{f},\ 0\leq e\leq f,\ \mathrm{gcd}(e,f)=1.\nonumber
\end{equation}
Then the integer $\lfloor\mathrm{log}_2(f+1)\rfloor$ is called the 2-adic complexity of the sequence $s$ and is denoted as $\Phi_{2}(s)$, i.e.,
\begin{equation}
\Phi_{2}(s)=\left\lfloor\mathrm{log}_2\left(\frac{2^N-1}{\mathrm{gcd}\left(2^N-1,S(2)\right)}+1\right)\right\rfloor,\label{2-adic calculation}
\end{equation}
where $\lfloor z\rfloor$ is the largest integer that is less than or equal to $z$.

It is well known that the 2-adic complexity of a binary sequence $s$ with period $N$ should be larger than $\frac{N}{2}$ to resist the Rational Approximation Algorithm by Klapper et al. \cite{Klapper-1}.

\section{The interleaved structures of a binary $m$-sequence and Yu-Gong sequence}\label{section 3}

$\mathbf{Construction\ 1}$ \cite{Yu-Gong-1}: Let each column of a $v\times w$ matrix $\mathbf{C}=(C_{i,j})$ be given by $C(i,j)=c_j(i)$ and $\mathbf{c}_{j}=(c_{j}(0),c_{j}(1),\cdots,c_{j}(v-1))$, $0\leq j\leq w-1$, i.e., the matrix $\mathbf{C}$ can be expressed as
\begin{equation}
\mathbf{C}=\left(
\begin{array}{cccc}
c_{0}(0) & c_{1}(0) &\cdots &c_{w-1}(0)\\
c_{0}(1) & c_{1}(1) &\cdots &c_{w-1}(1)\\
\vdots & \vdots &\ddots &\vdots\\
c_{0}(v-1) & c_{1}(v-1) &\cdots &c_{w-1}(v-1)
\end{array}
\right).\nonumber
\end{equation}
If each sequence $\mathbf{c}_{j}$ is either a cyclic shift of a binary sequence $\mathbf{a}=(a_{0},a_1,\cdots,a_{v-1})$ of period $v$ or a zero sequence and the sequence $\mathbf{u}=\{u_{t}\}$ is obtained by concatenating the successive rows of the above matrix $\mathbf{C}$, then $\mathbf{u}$ is called a $(v,w)$ interleaved sequence. By the definition, $\mathbf{c}_{j}=L^{e_{j}}(\mathbf{a})$, $0\leq j\leq w-1$, here $L^{e_{j}}$ is a cyclic $e_{j}$ left shift operation, $e_{j}\in \mathbb{Z}_{v}$ or $e_{j}=\infty$ if $\mathbf{c}_{j}$ is a zero sequence. Adding a binary sequence $\mathbf{b}=(b_0,b_1,\cdots,b_{w-1})$ of period $w$ to the sequence $\mathbf{u}$, a new sequence $s$ will be produced, which is denoted $\mathbf{s}:=I(\mathbf{a},\mathbf{e})+\mathbf{b}$, where $\mathbf{u}:=I(\mathbf{a},\mathbf{e})$ and $e=(e_0,e_1,\cdots,e_{w-1})$, and it still preserves the $(v,w)$ interleaved structure. We call $\mathbf{a}$, $\mathbf{e}$ and $\mathbf{b}$ the base, the shift and the indicator sequences of $\mathbf{s}$, respectively.

For a positive integer $k>1$, let $\mathbf{b}$ be a binary $m$-sequence of period $2^{2k}-1$, i.e., $w=2^{2k}-1$, where $b_t=Tr_1^{2k}(\alpha^t)$ and $\alpha$ is a primitive element of the finite field $\mathbb{F}_{2^{2k}}$, $0\leq t\leq 2^{2k}-2$. It is well known that $\mathbf{b}$ can be expressed as a $(2^k-1,2^k+1)$ interleaved sequence \cite{Gong-1}, i.e., $\mathbf{b}=I(\mathbf{a}^{\prime},\mathbf{e}^{\prime})$, where the base sequence $\mathbf{a}^{\prime}=(a_{0}^{\prime},a_1^{\prime},\cdots,a_{2^k-2}^{\prime})$ is a binary $m$-sequence of period $2^k-1$ defined by $a_i^{\prime}=Tr_1^{k}(\beta^{i})$, $0\leq i\leq 2^k-2$, $\beta=\alpha^{2^k+1}$ is a primitive element of $\mathbb{F}_{2^{k}}$, and the shift sequence $\mathbf{e^{\prime}}=(e_0^{\prime},e_1^{\prime},\cdots,e_{2^k}^{\prime})$ is given by $e_{0}^{\prime}=\infty$ and $\beta^{e_{j}^{\prime}}=Tr_k^{2k}(\alpha^{j})$ for $1\leq j\leq 2^k$.

A Yu-Gong sequence $\mathbf{s}=I(\mathbf{a},\mathbf{e})+\mathbf{b}$ of period $N=4(2^{2k}-1)$, is given by a $4\times(2^{2k}-1)$ interleaved structure, where:
\begin{itemize}
\item[(1)] $\mathbf{a}=(0,1,1,1)$ is the perfect binary sequence;
\item[(2)] $\mathbf{b}$ is the binary $m$-sequence defined as above;
\item[(3)] $\mathbf{e}$ is a sequence over $\mathbb{Z}_{4}$ represented as a $(2^k-1)\times(2^k+1)$ interleaved structure by a matrix $\mathbf{E}=(e_{i,j})$, where
\begin{eqnarray}
e_{i,j}=\left\{
\begin{array}{ll}
3i+\delta\ (\mathrm{mod}\ 4),\ \ \ \mathrm{if}\ j=0;\\
3(i+j)\ (\mathrm{mod}\ 4),\ \mathrm{if}\ 1\leq j\leq 2^k
\end{array}
\right.\label{shift-sequence}
\end{eqnarray}
for $0\leq i\leq 2^k-2$ and $\delta=1$ or $-1$.
\end{itemize}
\section{An interesting and useful law of the autocorrelation distribution of the Yu-Gong sequence}\label{section 2}

In order to analyze the 2-adic complexity of the Yu-Gong sequence $\mathbf{s}$, we need the exact order according to the value $\tau$ of the autocorrelation value $AC_{\mathbf{s}}(\tau)$ of $\mathbf{s}$, which can be given by the following two results.
\begin{theorem} \label{Yu-Gong}\cite{Yu-Gong-1}
Let $\mathbf{s}=I(\mathbf{a},\mathbf{e})+\mathbf{b}$ be a Yu-Gong sequence of period $N=4(2^{2k}-1)$, $k>1$. Then, it has the four-valued optimal autocorrelation of $AC_{\mathbf{s}}(\tau)\in\{N,\ 0,\ \pm4\}$ for any $\tau$. Precisely, its complete autocorrelation is given by
\begin{eqnarray}
AC_{\mathbf{s}}(\tau)=\left\{
\begin{array}{llllllll}
N,\ \ \ \ \ \ \ \ \ \ \ \mathrm{if}\ \tau=0\\
0,\ \ \ \ \ \ \ \ \ \ \ \ \mathrm{if}\ (\tau\neq0\ and\ x=0)\\
\ \ \ \ \ \ \ \ \ \ \ \ \ \ \mathrm{or}\ (x,y,v)=(\sigma,0,v)\\
\ \ \ \ \ \ \ \ \ \ \ \ \ \ \mathrm{or}\ (y,v)=(\psi,2),\\
-4,\ \ \ \ \ \ \ \ \ \mathrm{if}\ (x,y,v)=(\sigma,0,0)\\
\ \ \ \ \ \ \ \ \ \ \ \ \ \ \mathrm{or}\ (y,v)=(\psi,1)\\
\ \ \ \ \ \ \ \ \ \ \ \ \ \ \mathrm{or}\ (y,v)=(\psi,3),\\
+4,\ \ \ \ \ \ \ \ \ \mathrm{if}\ (y,v)=(\psi,0),
\end{array}
\right.\label{autocorrelation}
\end{eqnarray}
where $x\equiv \tau\ (\mathrm{mod}\ 2^{2k}-1),\ y\equiv \tau\ (\mathrm{mod}\ 2^k+1)$, and $v\equiv \tau\ (\mathrm{mod}\ 4)$. Also, $\sigma\in \mathbb{Z}_{2^{2k}-1}^{+},\ \psi\in \mathbb{Z}_{2^k+1}^{+}$, $v\in \mathbb{Z}_{4}^{+}$, and $\mathbb{Z}_{h}^{+}=\{1,2,\cdots,h-1\}$ for a positive integer $h$.
\end{theorem}
Now, the order of the autocorrelation value of $\mathbf{s}$ can be described as follows.

\begin{corollary}\label{main-lemma}
Let the symbols be the same as those in Theorem \ref{Yu-Gong}. Then the following results hold:
\begin{itemize}
\item[(1)]Suppose $1\leq\tau_1,\tau_2 \leq N-1$ and $\tau_1\equiv \tau_2\ (\mathrm{mod}\ 4(2^k+1))$. Then the autocorrelation function of the sequence $\mathbf{s}$ satisfies $AC_{\mathbf{s}}(\tau_1)=AC_{\mathbf{s}}(\tau_2)$. Particularly, $AC_{\mathbf{s}}(4(2^k+1)i)=-4$, $i=1,2,\cdots,2^k-2$;
\item[(2)]For $1\leq\tau\leq 4(2^k+1)$, if we divide the set of the autocorrelation values $\{AC_{\mathbf{s}}(\tau)\}_{\tau=1}^{4(2^k+1)}$ of $\mathbf{s}$ into $2^k+1$ subsets $S_j$'s, $j=1,2,\cdots,2^k+1$, according to the order of $\tau$ and each subset contains four elements, i.e., $S_j=\{AC_{\mathbf{s}}(4(j-1)+1),AC_{\mathbf{s}}(4(j-1)+2),AC_{\mathbf{s}}(4(j-1)+3),AC_{\mathbf{s}}(4(j-1)+4)\}$, Then
\begin{eqnarray}
S_1=S_2=\cdots=S_{2^{k-2}}=\{-4,0,-4,4\};\label{Auto-distribution-1}\\
S_{2^{k-2}+1}=\{0,0,-4,4\};\label{Auto-distribution-2}\\
S_{2^{k-2}+2}=S_{2^{k-2}+3}=\cdots=S_{3\times2^{k-2}}=\{-4,0,-4,4\};\label{Auto-distribution-3}\\
S_{3\times2^{k-2}+1}=\{-4,0,0,4\};\label{Auto-distribution-4}\\
S_{3\times2^{k-2}+2}=S_{3\times2^{k-2}+3}=\cdots=S_{2^{k}}=\{-4,0,-4,4\};\label{Auto-distribution-5}\\
S_{2^{k}+1}=\{-4,0,-4,-4\}.\label{Auto-distribution-6}
\end{eqnarray}
\end{itemize}
It should be pointed out that there are $2^{k-2}$ sets in Eq. (\ref{Auto-distribution-1}), $2^{k-1}-1$ sets in Eq. (\ref{Auto-distribution-3}), and $2^{k-2}-1$ sets in Eq.(\ref{Auto-distribution-5}).
\end{corollary}
\begin{proof}
Note that $\tau_1\equiv \tau_2\ (\mathrm{mod}\ 4(2^k+1))$ implies $\tau_1\equiv \tau_2\ (\mathrm{mod}\ 2^k+1)$ and $\tau_1\equiv \tau_2\ (\mathrm{mod}\ 4)$. Then, from Theorem \ref{Yu-Gong}, we find that $\tau_1$ and $\tau_2$ induce the same pair $(y,v)$, which leads to $AC_{\mathbf{s}}(\tau_1)=AC_{\mathbf{s}}(\tau_2)$.
Other results can also be directly verified by Theorem \ref{Yu-Gong}. \ \ \ \ \ \ \ \ \ \ \ \ \ \ \ \ \ \ \ \ \ \ \ \ \ \ \ \ \ \ \ \ \ \ \ \ \ \  \ \ \ \ \ \ \ \ \ \ \ \ \ \ \ \ \ \ \ \ \ \ \ $\Box$
\end{proof}
\textcolor{blue}{
\begin{example}
By direct computation using Matlab programs, the autocorrelation distributions of Yu-Gong sequences for $k=2,3,4$ have been listed in Tables 1-3. And we have marked these three special sets $S_{2^{k-2}+1},S_{3\times2^{k-2}+1},S_{2^{k}+1}$ in red. Especially, the number of sets in Eq. (\ref{Auto-distribution-5}) is $0$ since $2^{k-2}-1=0$ for $k=2$ and the autocorrelation distribution is the set $\{-4,0,0,4,-4,0,-4,4,0,0,-4\}$ for $k=1$ and $\tau=1,2,3,4,5,6,7,8,9,10,11$.
\end{example}}
\textcolor{blue}{
\begin{table}
\centering
{\tiny
\bf{\textcolor{blue}{Table 1: The autocorrelation of Yu-Gong sequence for $k=2$}}
\begin{tabular}{|c|c|}
\hline$\tau$&$\mathrm{AC}_\mathbf{s}(\tau)$\\
\hline $\ \ \ 1-20\ (S_{1}-S_{5})$&$-4,0,-4,4;\textcolor{red}{\mathbf{0,0,-4,4}};-4,0,-4,4;\textcolor{red}{\mathbf{-4,0,0,4}};\textcolor{red}{\mathbf{-4,0,-4,-4}};$\\
\hline $21-40\ \ \ \ \ \ \ \ \ \ \ \ $&$-4,0,-4,4;\textcolor{red}{\mathbf{0,0,-4,4}};-4,0,-4,4;\textcolor{red}{\mathbf{-4,0,0,4}};\textcolor{red}{\mathbf{-4,0,-4,-4}};$\\
\hline $41-59\ \ \ \ \ \ \ \ \ \ \ \ $&$-4,0,-4,4;\textcolor{red}{\mathbf{0,0,-4,4}};-4,0,-4,4;\textcolor{red}{\mathbf{-4,0,0,4}};\textcolor{red}{\mathbf{-4,0,-4\ \ \ \ }}$;\\
\hline
\end{tabular}
}
\end{table}}
\newcommand{\tabincell}[2]{\begin{tabular}{@{}#1@{}}#2\end{tabular}}
\begin{table}
\centering
{\tiny
\bf{\textcolor{blue}{Table 2: The autocorrelation of Yu-Gong sequence for $k=3$}}
\begin{tabular}{|c|c|}
\hline$\tau$&$\mathrm{AC}_\mathbf{s}(\tau)$\\
\hline $\ \ \ 1-36\ (S_{1}-S_{9})$&\tabincell{l} {$-4,0,-4,4;-4,0,-4,4;\textcolor{red}{\mathbf{0,0,-4,4}};-4,0,-4,4;-4,0,-4,4;$\\
$-4,0,-4,4;\textcolor{red}{\mathbf{-4,0,0,4}};-4,0,-4,4;\textcolor{red}{\mathbf{-4,0,-4,-4}};$}\\
\hline $37-72\ \ \ \ \ \ \ \ \ \ \ \ $&\tabincell{l} {$-4,0,-4,4;-4,0,-4,4;\textcolor{red}{\mathbf{0,0,-4,4}};-4,0,-4,4;-4,0,-4,4;$\\
$-4,0,-4,4;\textcolor{red}{\mathbf{-4,0,0,4}};-4,0,-4,4;\textcolor{red}{\mathbf{-4,0,-4,-4}};$}\\
\hline $73-108\ \ \ \ \ \ \ \ \ \ \ \ $&\tabincell{l} {$-4,0,-4,4;-4,0,-4,4;\textcolor{red}{\mathbf{0,0,-4,4}};-4,0,-4,4;-4,0,-4,4;$\\
$-4,0,-4,4;\textcolor{red}{\mathbf{-4,0,0,4}};-4,0,-4,4;\textcolor{red}{\mathbf{-4,0,-4,-4}};$}\\
\hline $109-144\ \ \ \ \ \ \ \ \ \ \ \ $&\tabincell{l} {$-4,0,-4,4;-4,0,-4,4;\textcolor{red}{\mathbf{0,0,-4,4}};-4,0,-4,4;-4,0,-4,4;$\\
$-4,0,-4,4;\textcolor{red}{\mathbf{-4,0,0,4}};-4,0,-4,4;\textcolor{red}{\mathbf{-4,0,-4,-4}};$}\\
\hline $145-180\ \ \ \ \ \ \ \ \ \ \ \ $&\tabincell{l} {$-4,0,-4,4;-4,0,-4,4;\textcolor{red}{\mathbf{0,0,-4,4}};-4,0,-4,4;-4,0,-4,4;$\\
$-4,0,-4,4;\textcolor{red}{\mathbf{-4,0,0,4}};-4,0,-4,4;\textcolor{red}{\mathbf{-4,0,-4,-4}};$}\\
\hline $181-216\ \ \ \ \ \ \ \ \ \ \ \ $&\tabincell{l} {$-4,0,-4,4;-4,0,-4,4;\textcolor{red}{\mathbf{0,0,-4,4}};-4,0,-4,4;-4,0,-4,4;$\\
$-4,0,-4,4;\textcolor{red}{\mathbf{-4,0,0,4}};-4,0,-4,4;\textcolor{red}{\mathbf{-4,0,-4,-4}};$}\\
\hline $217-251\ \ \ \ \ \ \ \ \ \ \ \ $&\tabincell{l} {$-4,0,-4,4;-4,0,-4,4;\textcolor{red}{\mathbf{0,0,-4,4}};-4,0,-4,4;-4,0,-4,4;$\\
$-4,0,-4,4;\textcolor{red}{\mathbf{-4,0,0,4}};-4,0,-4,4;\textcolor{red}{\mathbf{-4,0,-4\ \ \ \ }};$}\\
\hline
\end{tabular}
}
\end{table}

\begin{table}
\centering
{\tiny
\bf{\textcolor{blue}{Table 3: The autocorrelation of Yu-Gong sequence for $k=4$}}
\begin{tabular}{|c|c|}
\hline$\tau$&$\mathrm{AC}_\mathbf{s}(\tau)$\\
\hline $\ \ \ 1-68\ (S_{1}-S_{17})$&\tabincell{l} {$-4,0,-4,4;-4,0,-4,4;-4,0,-4,4;-4,0,-4,4;\textcolor{red}{\mathbf{0,0,-4,4}};$\\
$-4,0,-4,4;-4,0,-4,4;-4,0,-4,4;-4,0,-4,4;-4,0,-4,4;$\\
$-4,0,-4,4;-4,0,-4,4;\textcolor{red}{\mathbf{-4,0,0,4}};-4,0,-4,4;-4,0,-4,4;$\\
$-4,0,-4,4;\textcolor{red}{\mathbf{-4,0,-4,-4}};$}\\
\hline $69-136\ \ \ \ \ \ \ \ \ \ \ \ $&\tabincell{l} {$-4,0,-4,4;-4,0,-4,4;-4,0,-4,4;-4,0,-4,4;\textcolor{red}{\mathbf{0,0,-4,4}};$\\
$-4,0,-4,4;-4,0,-4,4;-4,0,-4,4;-4,0,-4,4;-4,0,-4,4;$\\
$-4,0,-4,4;-4,0,-4,4;\textcolor{red}{\mathbf{-4,0,0,4}};-4,0,-4,4;-4,0,-4,4;$\\
$-4,0,-4,4;\textcolor{red}{\mathbf{-4,0,-4,-4}};$}\\
\hline $137-204\ \ \ \ \ \ \ \ \ \ \ \ $&\tabincell{l} {$-4,0,-4,4;-4,0,-4,4;-4,0,-4,4;-4,0,-4,4;\textcolor{red}{\mathbf{0,0,-4,4}};$\\
$-4,0,-4,4;-4,0,-4,4;-4,0,-4,4;-4,0,-4,4;-4,0,-4,4;$\\
$-4,0,-4,4;-4,0,-4,4;\textcolor{red}{\mathbf{-4,0,0,4}};-4,0,-4,4;-4,0,-4,4;$\\
$-4,0,-4,4;\textcolor{red}{\mathbf{-4,0,-4,-4}};$}\\
\hline $205-272\ \ \ \ \ \ \ \ \ \ \ \ $&\tabincell{l} {$-4,0,-4,4;-4,0,-4,4;-4,0,-4,4;-4,0,-4,4;\textcolor{red}{\mathbf{0,0,-4,4}};$\\
$-4,0,-4,4;-4,0,-4,4;-4,0,-4,4;-4,0,-4,4;-4,0,-4,4;$\\
$-4,0,-4,4;-4,0,-4,4;\textcolor{red}{\mathbf{-4,0,0,4}};-4,0,-4,4;-4,0,-4,4;$\\
$-4,0,-4,4;\textcolor{red}{\mathbf{-4,0,-4,-4}};$}\\
\hline $273-340\ \ \ \ \ \ \ \ \ \ \ \ $&\tabincell{l} {$-4,0,-4,4;-4,0,-4,4;-4,0,-4,4;-4,0,-4,4;\textcolor{red}{\mathbf{0,0,-4,4}};$\\
$-4,0,-4,4;-4,0,-4,4;-4,0,-4,4;-4,0,-4,4;-4,0,-4,4;$\\
$-4,0,-4,4;-4,0,-4,4;\textcolor{red}{\mathbf{-4,0,0,4}};-4,0,-4,4;-4,0,-4,4;$\\
$-4,0,-4,4;\textcolor{red}{\mathbf{-4,0,-4,-4}};$}\\
\hline $341-408\ \ \ \ \ \ \ \ \ \ \ \ $&\tabincell{l} {$-4,0,-4,4;-4,0,-4,4;-4,0,-4,4;-4,0,-4,4;\textcolor{red}{\mathbf{0,0,-4,4}};$\\
$-4,0,-4,4;-4,0,-4,4;-4,0,-4,4;-4,0,-4,4;-4,0,-4,4;$\\
$-4,0,-4,4;-4,0,-4,4;\textcolor{red}{\mathbf{-4,0,0,4}};-4,0,-4,4;-4,0,-4,4;$\\
$-4,0,-4,4;\textcolor{red}{\mathbf{-4,0,-4,-4}};$}\\
\hline $409-476\ \ \ \ \ \ \ \ \ \ \ \ $&\tabincell{l} {$-4,0,-4,4;-4,0,-4,4;-4,0,-4,4;-4,0,-4,4;\textcolor{red}{\mathbf{0,0,-4,4}};$\\
$-4,0,-4,4;-4,0,-4,4;-4,0,-4,4;-4,0,-4,4;-4,0,-4,4;$\\
$-4,0,-4,4;-4,0,-4,4;\textcolor{red}{\mathbf{-4,0,0,4}};-4,0,-4,4;-4,0,-4,4;$\\
$-4,0,-4,4;\textcolor{red}{\mathbf{-4,0,-4,-4}};$}\\
\hline $477-544\ \ \ \ \ \ \ \ \ \ \ \ $&\tabincell{l} {$-4,0,-4,4;-4,0,-4,4;-4,0,-4,4;-4,0,-4,4;\textcolor{red}{\mathbf{0,0,-4,4}};$\\
$-4,0,-4,4;-4,0,-4,4;-4,0,-4,4;-4,0,-4,4;-4,0,-4,4;$\\
$-4,0,-4,4;-4,0,-4,4;\textcolor{red}{\mathbf{-4,0,0,4}};-4,0,-4,4;-4,0,-4,4;$\\
$-4,0,-4,4;\textcolor{red}{\mathbf{-4,0,-4,-4}};$}\\
\hline $545-612\ \ \ \ \ \ \ \ \ \ \ \ $&\tabincell{l} {$-4,0,-4,4;-4,0,-4,4;-4,0,-4,4;-4,0,-4,4;\textcolor{red}{\mathbf{0,0,-4,4}};$\\
$-4,0,-4,4;-4,0,-4,4;-4,0,-4,4;-4,0,-4,4;-4,0,-4,4;$\\
$-4,0,-4,4;-4,0,-4,4;\textcolor{red}{\mathbf{-4,0,0,4}};-4,0,-4,4;-4,0,-4,4;$\\
$-4,0,-4,4;\textcolor{red}{\mathbf{-4,0,-4,-4}};$}\\
\hline $613-680\ \ \ \ \ \ \ \ \ \ \ \ $&\tabincell{l} {$-4,0,-4,4;-4,0,-4,4;-4,0,-4,4;-4,0,-4,4;\textcolor{red}{\mathbf{0,0,-4,4}};$\\
$-4,0,-4,4;-4,0,-4,4;-4,0,-4,4;-4,0,-4,4;-4,0,-4,4;$\\
$-4,0,-4,4;-4,0,-4,4;\textcolor{red}{\mathbf{-4,0,0,4}};-4,0,-4,4;-4,0,-4,4;$\\
$-4,0,-4,4;\textcolor{red}{\mathbf{-4,0,-4,-4}};$}\\
\hline $681-748\ \ \ \ \ \ \ \ \ \ \ \ $&\tabincell{l} {$-4,0,-4,4;-4,0,-4,4;-4,0,-4,4;-4,0,-4,4;\textcolor{red}{\mathbf{0,0,-4,4}};$\\
$-4,0,-4,4;-4,0,-4,4;-4,0,-4,4;-4,0,-4,4;-4,0,-4,4;$\\
$-4,0,-4,4;-4,0,-4,4;\textcolor{red}{\mathbf{-4,0,0,4}};-4,0,-4,4;-4,0,-4,4;$\\
$-4,0,-4,4;\textcolor{red}{\mathbf{-4,0,-4,-4}};$}\\
\hline $749-816\ \ \ \ \ \ \ \ \ \ \ \ $&\tabincell{l} {$-4,0,-4,4;-4,0,-4,4;-4,0,-4,4;-4,0,-4,4;\textcolor{red}{\mathbf{0,0,-4,4}};$\\
$-4,0,-4,4;-4,0,-4,4;-4,0,-4,4;-4,0,-4,4;-4,0,-4,4;$\\
$-4,0,-4,4;-4,0,-4,4;\textcolor{red}{\mathbf{-4,0,0,4}};-4,0,-4,4;-4,0,-4,4;$\\
$-4,0,-4,4;\textcolor{red}{\mathbf{-4,0,-4,-4}};$}\\
\hline $817-884\ \ \ \ \ \ \ \ \ \ \ \ $&\tabincell{l} {$-4,0,-4,4;-4,0,-4,4;-4,0,-4,4;-4,0,-4,4;\textcolor{red}{\mathbf{0,0,-4,4}};$\\
$-4,0,-4,4;-4,0,-4,4;-4,0,-4,4;-4,0,-4,4;-4,0,-4,4;$\\
$-4,0,-4,4;-4,0,-4,4;\textcolor{red}{\mathbf{-4,0,0,4}};-4,0,-4,4;-4,0,-4,4;$\\
$-4,0,-4,4;\textcolor{red}{\mathbf{-4,0,-4,-4}};$}\\
\hline $885-952\ \ \ \ \ \ \ \ \ \ \ \ $&\tabincell{l} {$-4,0,-4,4;-4,0,-4,4;-4,0,-4,4;-4,0,-4,4;\textcolor{red}{\mathbf{0,0,-4,4}};$\\
$-4,0,-4,4;-4,0,-4,4;-4,0,-4,4;-4,0,-4,4;-4,0,-4,4;$\\
$-4,0,-4,4;-4,0,-4,4;\textcolor{red}{\mathbf{-4,0,0,4}};-4,0,-4,4;-4,0,-4,4;$\\
$-4,0,-4,4;\textcolor{red}{\mathbf{-4,0,-4,-4}};$}\\
\hline $953-1019\ \ \ \ \ \ \ \ \ \ \ \ $&\tabincell{l} {$-4,0,-4,4;-4,0,-4,4;-4,0,-4,4;-4,0,-4,4;\textcolor{red}{\mathbf{0,0,-4,4}};$\\
$-4,0,-4,4;-4,0,-4,4;-4,0,-4,4;-4,0,-4,4;-4,0,-4,4;$\\
$-4,0,-4,4;-4,0,-4,4;\textcolor{red}{\mathbf{-4,0,0,4}};-4,0,-4,4;-4,0,-4,4;$\\
$-4,0,-4,4;\textcolor{red}{\mathbf{-4,0,-4\ \ \ \ }};$}\\
\hline
\end{tabular}
}
\end{table}

\section{The 2-adic complexity of Yu-Gong sequence}

In order to derive a lower bound on the 2-adic complexity of the Yu-Gong sequence $\mathbf{s}$, we need employ the method of Hu \cite{Hu-1}. It can be described as the following Lemma \ref{Sun}, which have also been used in several other references \textcolor{blue}{\cite{Hofer-1,SunYuhua-1,SunYuhua-2,Xiong-2}.}
\begin{lemma}\label{Sun}\cite{Hu-1}
Let $\mathbf{s}=\big(s_0,s_1,\cdots,s_{N-1}\big)$ be a binary sequence of period $N$, $S(x)=\sum\limits_{i=0}^{N-1}s_{i}x^i\in \mathbb{Z}[x]$ and $T(x)=\sum\limits_{i=0}^{N-1}(-1)^{s_{i}}x^i\in \mathbb{Z}[x]$. Then
\begin{eqnarray}
-2S(x)T\big(x^{-1}\big)\equiv N+\sum\limits_{\tau=1}^{N-1}AC_{\mathbf{s}}(\tau)x^{\tau}-T\big(x^{-1}\big)\left(\sum\limits_{i=0}^{N-1}x^i\right)\pmod{x^N-1}.\nonumber\ \Box
\end{eqnarray}
\end{lemma}
\textcolor{blue}{
The following Lemma \ref{number-theory-1} is also important to prove our main result.
\begin{lemma}\label{number-theory-1}
Let $k$ be a positive integer. Then the following results hold:
\begin{itemize}
\item[(1)] For $k\equiv2\ \mathrm{mod}\ 4$, we have
\begin{eqnarray}
5\mid \mathrm{gcd}\Big(2^{2k-1}-2^k+1,\frac{2^{2(2^k+1)}+1}{5}\Big).\label{Extend-0}
\end{eqnarray}
But for $k\equiv0\ (\mathrm{mod}\ 4)$, we have
\begin{eqnarray}
5\nmid \mathrm{gcd}\Big(2^{2k-1}-2^k+1,\frac{2^{2(2^k+1)}+1}{5}\Big).\label{Extend-000}
\end{eqnarray}
\item[(2)]
$
\mathrm{gcd}\Big(2^{2k-1}-2^k+1,\frac{2^{2(2^k+1)}+1}{5}\Big)
\left\{
\begin{array}{ll}
=1,\ \ \ \ \ \ \mathrm{if}\ 2^{2k-1}-2^k+1\ \mathrm{is\ a}\\
\ \ \ \ \ \ \ \ \ \ \ \mathrm{\ prime\ for\ }k\equiv0\ (\mathrm{mod}\ 4),\\
< 2^{2k-1},\ \mathrm{otherwise.}
\end{array}
\right.
$
\item[(3)] $\mathrm{gcd}\Big(2^{k-1}-1,2^{2^k+1}+1\Big)\left\{
\begin{array}{llll}
=1,\ \ \ \ \ \mathrm{if}\ k\ \mathrm{is\ even},\\
< 2^{k-1},\ \mathrm{otherwise}.
\end{array}
\right.$
\end{itemize}
\end{lemma}
\begin{proof}
\begin{itemize}
\item[(1)] Suppose $k=4t+2$. Note that the multiplicative order of $2$ modular $5$ is 4. On one hand, we have
$$2^{2k-1}-2^k+1=2^{8t+3}-2^{4t+2}+1\equiv2^3-2^2+1\equiv0\ (\mathrm{mod}\ 5),$$
i.e.,
\begin{eqnarray}
5\mid(2^{2k-1}-2^k+1),\label{extend-00}
\end{eqnarray}
on the other hand, since $5\mid (2^{4t}-1)$, we get $20\mid [4(2^{4t}-1)]$, i.e., $20\mid (2^k-4)$, which implies $2^k+1\equiv5\ (\mathrm{mod}\ 20)$ and $2(2^k+1)\equiv10\ (\mathrm{mod}\ 20)$. It is easy to know that the multiplicative order of $2$ modular $25$ is $20$. Then we have $2^{2(2^k+1)}+1\equiv2^{10}+1\equiv0\ (\mathrm{mod}\ 25)$, i.e.,
\begin{eqnarray}
5\mid \frac{2^{2(2^k+1)}+1}{5}. \label{extend-01}
\end{eqnarray}
Combining (\ref{extend-00}) and (\ref{extend-01}) we know that (\ref{Extend-0}) holds. But, if $k=4t$, we have $2^{2k-1}-2^{k}+1\equiv2^3-1+1\equiv3\ (\mathrm{mod}\ 5)$. Therefore, (\ref{Extend-000}) holds.
\item[(2)] Suppose that $2^{2k-1}-2^k+1$ is a prime. Then by little Fermat Theorem we have
\begin{eqnarray}
(2^{2k-1}-2^k+1)\mid (2^{2^{2k-1}-2^k}-1).\label{extend-1}
\end{eqnarray}
Note that
\begin{eqnarray}
\frac{2^{2(2^k+1)}+1}{5}\mid (2^{2(2^k+1)}+1),\ (2^{2(2^k+1)}+1)\mid (2^{4(2^k+1)}-1).\label{extend-2}
\end{eqnarray}
Furthermore, we have
\begin{eqnarray}
\mathrm{gcd}\Big(2^{2^{2k-1}-2^k}-1,2^{4(2^k+1)}-1\Big)&=&2^{\mathrm{gcd}\big(2^{2k-1}-2^k,4(2^k+1)\big)}-1\nonumber\\
&=&2^{\mathrm{gcd}\big(2^k(2^{k-1}-1),4(2^k+1)\big)}-1.\label{extend-3}
\end{eqnarray}
Without loss of generality, let $k\geq2$ (for $k=1$, the conclusion is trivial). Then
\begin{eqnarray}
\mathrm{gcd}\Big(2^k(2^{k-1}-1),4(2^k+1)\Big)&=&4\mathrm{gcd}\Big(2^{k-1}-1,2^k+1\Big)\nonumber\\
&=&4\mathrm{gcd}\big(2^{k-1}-1,3\big)\label{extend-4}\\
&=&\left\{
\begin{array}{ll}
4,\ \ \mathrm{if}\ k\ \mathrm{is\ even},\\
12,\ \mathrm{otherwise},
\end{array}
\right.\label{extend-5}
\end{eqnarray}
where (\ref{extend-4}) holds because $2^k+1=2(2^{k-1}-1)+3$.
Thus, by (\ref{extend-3}), (\ref{extend-5}), we get
\begin{eqnarray}
\mathrm{gcd}\Big(2^{2^{2k-1}-2^k}-1,2^{4(2^k+1)}-1\Big)=\left\{
\begin{array}{ll}
2^4-1,\ \ \mathrm{if}\ k\ \mathrm{is\ even},\\
2^{12}-1,\ \mathrm{otherwise},
\end{array}
\right.\label{extend-6}
\end{eqnarray}
We know that
\begin{eqnarray}
&2^{4(2^k+1)}-1=(2^{2(2^k+1)}-1)(2^{2(2^k+1)}+1),\label{extend-7}\\
&\mathrm{gcd}(2^{2(2^k+1)}-1,2^{2(2^k+1)}+1)=1,\label{extend-8}\\
&2^4-1=(2^2-1)(2^2+1)=3\times5,\label{extend-9}\\
&3\mid(2^{2(2^k+1)}-1),\ 3\nmid(2^{2(2^k+1)}+1),\label{extend-10}\\
&5\mid(2^{2(2^k+1)}+1).\label{extend-11}
\end{eqnarray}
Therefore, by (\ref{extend-6})-(\ref{extend-11}), we have
\begin{eqnarray}
\mathrm{gcd}\Big(2^{2^{2k-1}-2^k}-1,2^{2(2^k+1)}+1\Big)=5\ \mathrm{for\ an\ even}\ k.\label{extend-12}
\end{eqnarray}
Finally, combining (\ref{Extend-0})-(\ref{Extend-000}), (\ref{extend-1})-(\ref{extend-2}) and (\ref{extend-12}), the result follows.
\item[(3)] It is easy to known that
\begin{eqnarray}
&&2^{2(2^k+1)}-1=(2^{2^k+1}+1)(2^{2^k+1}-1),\nonumber\\
&&\mathrm{gcd}\Big(2^{2^k+1}+1,2^{2^k+1}-1\Big)=1,\nonumber\\
&&\mathrm{gcd}\Big(2^{k-1}-1,2^{2(2^k+1)}-1\Big)\nonumber\\
&&=\mathrm{gcd}\Big(2^{k-1}-1,2^{2^k+1}+1\Big)\times\mathrm{gcd}\Big(2^{k-1}-1,2^{2^k+1}-1\Big),\label{extend-17}\\
&&=2^{\mathrm{gcd}\big(k-1,2(2^k+1)\big)}-1,\label{extend-18}\\
&&\mathrm{gcd}\Big(2^{k-1}-1,2^{2^k+1}-1\Big)=2^{\mathrm{gcd}\big(k-1,2^k+1\big)}-1.\label{extend-19}
\end{eqnarray}
For an even $k$, $k-1$ is odd, then
$$\mathrm{gcd}\big(k-1,2(2^k+1)\big)=\mathrm{gcd}\big(k-1,2^k+1\big),$$
which results in
\begin{eqnarray}
&&\mathrm{gcd}\Big(2^{k-1}-1,2^{2(2^k+1)}-1\Big)=\mathrm{gcd}\Big(2^{k-1}-1,2^{2^k+1}-1\Big)\label{extend-19}
\end{eqnarray}
by (\ref{extend-18})-(\ref{extend-19}). Furthermore, combining (\ref{extend-17}) and (\ref{extend-19}), we can get the desired result.
\end{itemize}
\end{proof}
}
\textcolor{blue}{
\begin{example}
By direct computation using Mathematica programs, we find that the smallest two positive integers $k$'s such that $2^{2k-1}-2^k+1$ are primes and $k\equiv0\ (\mathrm{mod}\ 4)$ are $4$, $24$.
\end{example}
}
Now, we present our main result.
\textcolor{blue}{
\begin{theorem}\label{main result}
Let $\mathbf{s}=I(\mathbf{a},\mathbf{e})+\mathbf{b}$ be the Yu-Gong sequence of period $N=4(2^{2k}-1)$, $k>1$, introduced in Section \ref{section 3}. Then the 2-adic complexity \textcolor{blue}{$\Phi_2(\mathbf{s})$ of $\mathbf{s}$} satisfies the following lower bound
\begin{eqnarray}
\Phi_2(\mathbf{s})\left\{
\begin{array}{llll}
=N,\ \ \ \ \ \ \ \ \ \ \ \ \ \ \ \ \ \ \mathrm{if}\ k\equiv0\ (\mathrm{mod}\ 4)\ \mathrm{and}\ (2^{2k-1}-2^{k}+1)\ \mathrm{is\ a\ prime,}\\
> N-\mathrm{log}_2N+1,\ \ \mathrm{if}\ k\ \mathrm{is\ even,}\\
> N-2\mathrm{log}_2N+4,\ \mathrm{otherwise,}\\
\end{array}
\right.\nonumber
\end{eqnarray}
i.e., the 2-adic complexity of $\mathbf{s}$ far outweight one half of the period.\ \ \ \ \ \ \ \ \ \ $\Box$
\end{theorem}
}
\textcolor{blue}{
\begin{proof} Above all, by Lemma \ref{Sun}, we know that
\begin{eqnarray}
S(2)T(2^{-1})\equiv-\frac{1}{2}\sum_{\tau=1}^{4(2^{2k}-1)-1}AC_{\mathbf{s}}(\tau)2^{\tau}-2(2^{2k}-1)\ \Big(\mathrm{mod}\ 2^{4(2^{2k}-1)}-1\Big)\nonumber
\end{eqnarray}
\textcolor{blue}{From Corollary \ref{main-lemma} (1)}, if we add the value $-4$ to the end of the sequence $\{AC_{\mathbf{s}}(\tau)\}_{\tau=1}^{4(2^{2k}-1)}$, we will get a sequence segment consisting $2^k-1$ periods of the sequence $\{AC_{\mathbf{s}}(\tau)\}_{\tau=1}^{4(2^k+1)}$. Therefore, we have
\begin{eqnarray}
S(2)T(2^{-1})&\equiv&-\frac{1}{2}\Big\{\sum_{\tau=1}^{4(2^{2k}-1)-1}AC_s(\tau)2^{\tau}+(-4)\times 2^{4(2^{2k}-1)}-(-4)\times 2^{4(2^{2k}-1)}\Big\}\nonumber\\
&&-2(2^{2k}-1)\ \Big(\mathrm{mod}\ 2^{4(2^{2k}-1)}-1\Big)\nonumber\\
&=&-\frac{1}{2}\Big\{\Big(\sum_{i=0}^{2^k-2}2^{4(2^k+1)i}\Big)\Big(\sum_{\tau=1}^{4(2^k+1)}AC_{\mathbf{s}}(\tau)2^{\tau}\Big)+\textcolor{blue}{4\times2^{4(2^{2k}-1)}}\Big\}\nonumber\\
&&-2(2^{2k}-1)\ \Big(\mathrm{mod}\ 2^{4(2^{2k}-1)}-1\Big)\nonumber\\
&\equiv&-\frac{1}{2}\Big\{\Big(\frac{2^{4(2^{2k}-1)}-1}{2^{4(2^k+1)}-1}\Big)\Big(\sum_{\tau=1}^{4(2^k+1)}AC_s(\tau)2^{\tau}\Big)+4\Big\}\nonumber\\
&&-2(2^{2k}-1)\ \Big(\mathrm{mod}\ 2^{4(2^{2k}-1)}-1\Big)\nonumber\\
&=&-\frac{1}{2}\Big(\frac{2^{4(2^{2k}-1)}-1}{2^{4(2^k+1)}-1}\Big)\Big(\sum_{\tau=1}^{4(2^k+1)}AC_s(\tau)2^{\tau}\Big)\nonumber\\
&&-2^{2k+1}\ \Big(\mathrm{mod}\ 2^{4(2^{2k}-1)}-1\Big)\label{Key-equation}
\end{eqnarray}
\textcolor{blue}{From Corollary \ref{main-lemma} (2)}, we have
\begin{eqnarray}
&&\sum_{\tau=1}^{4(2^k+1)}AC_s(\tau)2^{\tau}=\Big(\sum_{i=0}^{2^k}2^{4i}\Big)\Big((-4)\times2+0\times2^2+(-4)\times2^3+4\times2^4\Big)\nonumber\\
&&-2^{4\times2^{k-2}}\Big((-4)\times2+0\times2^2+(-4)\times2^3+4\times2^4\Big)\nonumber\\
&&+2^{4\times2^{k-2}}\Big(0\times2+0\times2^2+(-4)\times2^3+4\times2^4\Big)\nonumber\\
&&-2^{4\times3\times2^{k-2}}\Big((-4)\times2+0\times2^2+(-4)\times2^3+4\times2^4\Big)\nonumber\\
&&+2^{4\times3\times2^{k-2}}\Big((-4)\times2+0\times2^2+0\times2^3+4\times2^4\Big)\nonumber\\
&&-2^{4\times2^{k}}\Big((-4)\times2+0\times2^2+(-4)\times2^3+4\times2^4\Big)\nonumber\\
&&+2^{4\times2^{k}}\Big((-4)\times2+0\times2^2+(-4)\times2^3+(-4)\times2^4\Big)\nonumber\\
&&=4\times6\times\frac{6(2^{4(2^k+1)}-1)}{2^4-1}+4\times2^{2^k+1}+4\times2^{3\times(2^k+1)}-8\times2^{4\times(2^k+1)}\nonumber\\
&&=8\Big\{\frac{3(2^{4(2^k+1)}-1)}{2^4-1}+2^{2^k}+2^{3\times2^k+2}-2^{4\times(2^k+1)}\Big\}.\label{key-2}
\end{eqnarray}
Bringing (\ref{key-2}) into (\ref{Key-equation}) and simplifying it, we get
\begin{eqnarray}
S(2)T(2^{-1})&\equiv&-4\Bigg\{\frac{2^{4(2^{2k}-1)}-1}{2^{4(2^k+1)}-1}\Big[\frac{2^{4(2^k+1)}-1}{5}+2^{2^k}+2^{3\times2^k+2}\nonumber\\
&&\ \ \ \ \ \ \ -2^{4\times(2^k+1)}\Big]+2^{2k-1}\Bigg\}\ \Big(\mathrm{mod}\ 2^{4(2^{2k}-1)}-1\Big).\label{Key-3}
\end{eqnarray}
On one hand, it is obvious that $2^4-1=3\times5$ is a factor of $2^{4(2^{2k}-1)}-1$, on the other hand, we can derive
\begin{eqnarray}
S(2)T(2^{-1})&\equiv&\left\{
\begin{array}{llll}
13\ (\mathrm{mod}\ 15),\ \ \ \mathrm{if}\ k\equiv0\ (\mathrm{mod}\ 4),\\
0\ \ \ (\mathrm{mod}\ 15),\ \ \ \mathrm{if}\ k\equiv1\ (\mathrm{mod}\ 4),\\
10\ (\mathrm{mod}\ 15),\ \ \ \mathrm{if}\ k\equiv2\ (\mathrm{mod}\ 4),\\
9\ \ \ (\mathrm{mod}\ 15),\ \ \ \mathrm{if}\ k\equiv3\ (\mathrm{mod}\ 4)
\end{array}
\right.
\label{Key-7}
\end{eqnarray}
from (\ref{Key-3}) by direct calculation. This tells us that
\begin{eqnarray}
\mathrm{gcd}\Big(S(2)T(2^{-1}),2^{4(2^{2k}-1)}-1\Big)>1\ \ \mathrm{for}\ k\equiv1,2,3\ (\mathrm{mod}\ 4).
\end{eqnarray}
In order to obtain a more exact lower bound on the 2-adic complexity of Yu-Gong sequence, we need to give some more detailed computation.
Again from (\ref{Key-3}), we get
\begin{eqnarray}
S(2)T(2^{-1})&\equiv&\textcolor{blue}{-2^{2k+1}}\ \Bigg(\mathrm{mod}\ \frac{2^{4(2^{2k}-1)}-1}{2^{4(2^k+1)}-1}\Bigg),\label{Key-4}\\
S(2)T(2^{-1})&\equiv&-4\Bigg\{(2^k-1)\Big[\frac{2^{4(2^k+1)}-1}{5}+2^{2^k}+2^{3\times2^k+2}-1\Big]\nonumber\\
&&\ \ \ \ \ \ \ \ +\textcolor{blue}{2^{2k-1}}\Bigg\}\ \Big(\mathrm{mod}\ 2^{4(2^k+1)}-1\Big),\label{Key-5}
\end{eqnarray}
where (\ref{Key-5}) comes from the following congruence
$$\frac{2^{4(2^{2k}-1)}-1}{2^{4(2^k+1)}-1}=\frac{2^{4(2^k+1)(2^k-1)}-1}{2^{4(2^k+1)}-1}\equiv\ 2^k-1\ \Big(\mathrm{mod}\ 2^{4(2^k+1)}-1\Big).$$
Furthermore, since $2^{4(2^k+1)}-1=5\times\frac{2^{2(2^k+1)}+1}{5}\times(2^{2^k+1}+1)\times(2^{2^k+1}-1)$, then by (\ref{Key-5}) we have
\begin{eqnarray}
S(2)T(2^{-1})&\equiv&-4(2^{2k-1}-2^k+1)\ \Big(\mathrm{mod}\ \frac{2^{2(2^k+1)}+1}{5}\Big),\label{Key-8}\\
S(2)T(2^{-1})&\equiv&\textcolor{blue}{-8(2^{k-1}-1)^{2}}\ \Big(\mathrm{mod}\ 2^{2^k+1}+1\Big),\label{Key-9}\\
S(2)T(2^{-1})&\equiv&-2^{2k+1}\ \Big(\mathrm{mod}\ 2^{2^k+1}-1\Big).\label{Key-10}
\end{eqnarray}
Combining the results in Lemma \ref{number-theory-1}, the proof is finished.
\end{proof}
}
\textcolor{blue}{
\begin{example}
To ensure the correctness of our main result, at the same time, in order to compare the actual values with the lower bounds of the 2-adic complexity of Yu-Gong sequences obtained in this paper, we have done the following verification work by combining Matlab and Mathematica programs:
\begin{itemize}
\item[(1)] For $k=1,2,3,4,5,6,7,8$, the correctness of the congruences (\ref{Key-8})-(\ref{Key-10}) have been verified using the direct definitions of Yu-Gong sequences and the mathematical expression $S(2)T(2^{-1})$.
\item[(2)] For $k=1,2,3,4,5,6,7,8$, the actual values of the 2-adic complexity of Yu-Gong sequences have been determined by determining the corresponding $\mathrm{gcd}(S(2),2^N-1)$ in the definition of 2-adic complexity of binary sequences. And we list a table to compare the actual values and the lower bounds of the 2-adic complexity of Yu-Gong sequences (Please see Table 4).
\item[(3)] From the process of determining the lower bound of the 2-adic complexity of Yu-Gong sequence, the value of $\mathrm{gcd}(2^{2k-1}-2^k+1, \frac{2^{2(2^k+1)}+1}{5})$ is a key factor, which especially affects the cases of the maximal values of the 2-adic complexity. In Lemma \ref{number-theory-1}, we proved $\mathrm{gcd}(2^{2k-1}-2^k+1, \frac{2^{2(2^k+1)}+1}{5})=1$ only when $k\equiv0\ \mathrm{mod}\ 4$ and $2^{2k-1}-2^k+1$ is a prime. In fact, we find
    \begin{eqnarray}
    \mathrm{gcd}\Big(2^{2k-1}-2^k+1, \frac{2^{2(2^k+1)}+1}{5}\Big)=1\label{conjecture}
   \end{eqnarray}
     for each $k\in\{4,8,12,16,20,24,28,32\}$ by Mathematica programs and we can not determine it for the cases of $k\equiv0\ \mathrm{mod}\ 4$ and $k\geq36$ because of the limitation of computer performance. So we guess that (\ref{conjecture}) maybe hold for all $k\equiv0\ (\mathrm{mod}\ 4)$. But it is difficult for us within our capabilities to prove it now. We also sincerely invite interested readers to complete it.
\end{itemize}
\end{example}
}
\begin{table}
\centering
{\tiny
\bf{\textcolor{blue}{Table 4: A comparison between the actual values and\\the lower bounds of the 2-adic complexity of Yu-Gong sequences}}
\begin{tabular}{|c|c|c|c|}
\hline The value of $k$ & \tabincell{l} {The period $N$ of\\ Yu-Gong sequence} & \tabincell{l} {The actual value of\\ the 2-adic complexity\\ of Yu-Gong sequence} & \tabincell{l} {The lower bound of\\ the 2-adic complexity\\ of Yu-Gong sequence\\ obtained in this paper}\\
\hline 1 & 12 & 8 & 6\\
\hline 2 & 60 & 60& 55\\
\hline 3 & 252 & 250& 240\\
\hline 4 & 1020 & 1020& 1020\\
\hline 5 & 4092 & 4082 & 4072\\
\hline 6 & 16380 & 16380 & 16367\\
\hline 7 & 65532 & 65530 & 65504\\
\hline 8 & 262140 & 262140 & 262123\\
\hline
\end{tabular}
}
\end{table}

\end{document}